\documentclass[a4paper,leqno]{amsart}

\usepackage{amsmath}
\usepackage{amsthm}
\usepackage{amssymb}
\usepackage{amsfonts}
\usepackage[applemac]{inputenc}
\usepackage{mathrsfs}
\usepackage{pdfsync}
\usepackage{color}
\usepackage{bbm}

\usepackage{palatino}

\def\C{{\mathbb C}}
\def\R{{\mathbb R}}
\def\N{{\mathbb N}}
\def\H{{\mathcal H}}
\def\D{{\mathcal D}}
\def\a{\mathfrak a}

\def\F{\mathcal F}
\def\le{\leqslant}
\def\ge{\geqslant}

\newcommand{\eps}{\varepsilon}

\DeclareMathOperator*{\essinf}{ess\,inf}

\DeclareMathOperator{\supp}{supp}

\theoremstyle{plain}
\newtheorem{theorem}{Theorem}[section]
\newtheorem{lemma}[theorem]{Lemma}

\newtheorem{proposition}[theorem]{Proposition}
\newtheorem{assumption}[theorem]{Assumption}

\theoremstyle{definition}

\newtheorem{remark}[theorem]{Remark}
\newtheorem*{remark*}{Remark}

\numberwithin{equation}{section}

\begin{document}

\title[Effective Dynamics in the High Field Limit]
{Effective Dynamics of Translationally Invariant Magnetic Schr\"odinger Equations in the High Field Limit}

\author[G. Nenciu]{Gheorghe Nenciu}

\address[G.~Nenciu]
{Institute of Mathematics ``Simion Stoilow'' of the Romanian Academy, 21, Calea Grivi\c{t}ei, 010702-Bucharest, Sector 1, Romania}
\email{gheorghe.nenciu@imar.ro}

\author[E. Richman]{Evelyn Richman}

\address[E.~Richman]
{Department of Mathematics and Computer Science, University of Puget Sound, 1500 N. Warner Street,
Tacoma, WA 98416-1088, USA}
\email{evelynrichman@pugetsound.edu}

\author[C. Sparber]{Christof Sparber$^\ast$}

\address[C.~Sparber]
{Department of Mathematics, Statistics, and Computer Science, M/C 249, University of Illinois at Chicago, 851 S. Morgan Street, Chicago, IL 60607, USA}
\email{sparber@math.uic.edu$^\ast$}

\begin{abstract}
We study the large field limit in Schr\"odinger equations with magnetic vector potentials describing 
translationally invariant $B$-fields with respect to the $z$-axis. 
In a first step, using regular perturbation theory, we derive an approximate description of the solution, provided 
the initial data is compactly supported in the Fourier-variable dual to $z\in \R$.  
The effective dynamics is thereby seen to produce high-frequency oscillations and large magnetic drifts.  
In a second step we show, by using the theory of almost invariant subspaces, that this asymptotic description is stable under polynomially bounded
perturbations that vanish in the vicinity of the origin. 
\end{abstract}

\date{\today}

\subjclass[2020]{81Q15}
\keywords{Schr\"odinger equation, magnetic confinement, perturbation theory, almost invariant subspaces}


\maketitle


\section{Introduction}\label{sec:intro}

In this work we study a class of linear Schr\"odinger equations which describe the behavior of charged, spinless 
particles under the influence of {\it translationally invariant magnetic fields} $\mathcal B(x)=\nabla \times \mathcal A(x)$ in $\R^3$. More specifically, 
we denote the spatial variables by $(x_1,x_2,z)\in \R^3$, with $x\equiv (x_1,x_2)$, and consider magnetic vector potentials of the form 
\begin{equation*}\label{A}
\mathcal A (x)= \big(0, 0,  A(x)\big).
\end{equation*}
Note that this implies that $\mathcal A(x)$ is automatically divergence free, i.e. $\nabla \cdot \mathcal A(x)=0$. 
Prototypical examples include the field around an infinite current-carrying wire,
\[
\mathcal A = \big(0, 0, b\ln(|x|)\big),\quad b\in \R,
\]
and azimuthal fields of constant magnitude,
\begin{equation}\label{Aex}
 \mathcal A = \big(0, 0, b |x|\big),\quad b \in \R.
\end{equation}
The dynamics of quantum particles in such translationally invariant fields is governed by 
\begin{equation}\label{IVP}
	i\partial_t \Psi = \mathcal H \Psi , \quad \Psi_{\mid t=0} = \Psi _0\in L^2(\R^3),
\end{equation}
with (purely) magnetic Hamiltonian
\[
	\H = (-i\nabla + \mathcal A(x))^2.
\]
As previously observed in \cite{yafaev1, yafaev2}, the operator $\mathcal H$ can be treated as a family of two-dimensional operators under the partial Fourier transform 
$\mathcal F:L^2(\R)\to L^2(\R)$, given by
\[
	\F\varphi (x,p) \equiv \widehat\varphi (x,p) = \frac{1}{\sqrt{2\pi}}\int_\R e^{-izp}\varphi (x,z)\, dz, \quad \varphi \in L^2(\R^3).
\]
The translation invariance of the magnetic vector potential then implies that $\mathcal H$ is unitarily equivalent to the direct integral in $L^2(\R_p;L^2(\R^2_x))$ of the family of fiber Hamiltonians $\mathcal H(p)$, i.e.
\[
\mathcal F \mathcal H \mathcal F^{-1} = \int_\R^\oplus \mathcal H(p) \, dp.
\]
Within each fiber, the dynamics is then governed by 
\begin{equation}\label{IVPf}
	i\partial_t \widehat \Psi  =  \mathcal H(p) \widehat \Psi , \quad \widehat \Psi _{\mid t=0} = \widehat \Psi _0(\cdot, p).
\end{equation}
where, for fixed $p\in \R$, $\mathcal H(p)$ is an operator acting on $L^2(\R^2)$ only. We note that the hereby obtained fiber Hamiltonian
\[
\mathcal H(p) = -\Delta_x + (p+A(x))^2
\]
is symmetric and positive on $C_0^\infty(\R^2)$, where throughout this paper we shall always impose the condition $A \in L^\infty_{\rm loc}(\R^2)$. A classical 
result in the mathematical theory of Schr\"odinger operators (see, e.g., \cite{MaSh}) then implies that $\mathcal H(p)$ is essentially self-adjoint on $L^2(\R^2)$ and we 
shall denote its unique self-adjoint extension by the same letter. By Stone's theorem, $\mathcal H(p)$ is seen to be the generator of a group of unitary operators
$$
\mathcal S_t(p)=e^{-it \mathcal H(p)}: \ L^2(\R^2)\to L^2(\R^2).
$$
which yields the existence of a unique solution to \eqref{IVPf}. In the following, we will be interested in an effective description of this type of quantum dynamics 
in the {\it high magnetic field limit}, i.e., the regime in which the field strength $|\mathcal B|\to +\infty$. 


\subsection{Asymptotic regime and rescaling} 


We first consider the case of 
{\it homogenous} vector potentials, which can be seen as a generalization of \eqref{Aex}. More specifically, we consider
\begin{equation}\label{homA}
	\mathcal A^\eps(x) = \frac{1}{\eps^{\alpha +1}} \big(0,0,|x|^\alpha \Theta(\vartheta)\big), \quad \text{with $\alpha > 0$},
\end{equation}
where $|x|\ge 0$ and $\vartheta \in [0, 2\pi)$ denote polar coordinates in $\R^2$. The function $\Theta\in C_{\rm per}([0, 2\pi); \R)$ describes 
a possible azimuthal dependence of the magnetic field. Throughout this work, we shall impose: 

\begin{assumption}\label{assump}
The function $\Theta\in C_{\rm per}([0, 2\pi); \R)$ satisfies
\begin{equation*}
0 < c_1 \le \Theta(\vartheta) \le c_2 < \infty, \quad \forall \vartheta\in [0, 2\pi).
\end{equation*}
\end{assumption}

In addition, the small dimensionless parameter $0<\eps \ll 1$ describes the inverse of the magnetic field strength $|\mathcal B^\eps|$, 
and we are consequently interested in deriving the effective dynamics of solutions to \eqref{IVPf} and \eqref{homA} as $\eps \to 0_+$. The power $\alpha+1$ appearing 
in \eqref{homA} might seem ad hoc at first glance, but it is chosen in a way that will allow us to expand various quantities in integer powers of the 
perturbation parameter $(\eps p)$, see below. In principle, other powers can be considered as well, but the 
formulas for the expansion become much more cumbersome. 
In view of the fact that $\mathcal A^\eps$ can be rewritten as
\[
\mathcal A^\eps(x) = \frac{1}{\eps} \Big(0,0,\Big|\frac{x}{\eps}\Big|^\alpha \Theta(\vartheta)\Big) = \frac{1}{\eps}\Big(0,0, A\Big(\frac{x}{\eps}\Big)\Big)
\]
it is natural to introduce rescaled spatial variables via
\begin{equation*}\label{res}
	x' = \frac{x}{\eps} , \quad z' = z.
\end{equation*}
The new unknown is then given by 
\[
\psi^\eps(t, x', z') = \eps \Psi(t, \eps  x', z'),
\]
where the prefactor ensures that the $L^2$-norm remains invariant. From now on, 
we shall also assume that the initial data $\Psi_0$ is $\eps$-{\it independent} in the rescaled variables $(x',z')$ and normalized s.t. $\| \psi_0 \|_{ L^2} = 1$.  
\begin{remark} In terms of the original variables, this means that 
\[
\Psi_0 (x,z) = \frac{1}{\eps} \psi_0 \left(\frac{x}{\eps}, z \right). 
\]
In other words, $\Psi_0$ is concentrated on the scale $\eps $ w.r.t. to the $x$-directions, an assumption consistent with the asymptotic regime considered.
\end{remark} 
Rescaling the initial value problem \eqref{IVP} accordingly, we obtain (after dropping all the primes $'$, for simplicity)
\begin{equation}\label{resIVP}
i\eps^2  \partial_t \psi^\eps = \mathcal H^\eps \psi^\eps, \quad \psi^\eps_{\mid t=0} = \psi _0(x,z),
\end{equation}
where the rescaled Hamiltonian reads
\begin{equation}\label{resH}
\mathcal H^\eps =  - \Delta_x  - \eps^2 \partial_z^2 + |x|^{2\alpha}\Theta(\vartheta)^2 - 2i\eps|x|^\alpha \Theta(\vartheta)\partial_z. 
\end{equation}
We seek an asymptotic description as $\eps \to 0_+$ of 
\[
\psi^\eps(t,x,z)=e^{-it\H^\eps/\eps^2}\psi_0(x,z),
\]
solution to \eqref{resIVP}. To this end, we again have the fiber decomposition
\[
\mathcal F \mathcal H^\eps \mathcal F^{-1} = \int_\R^\oplus \mathcal H^\eps(p) \, dp.
\]
where $\mathcal H^\eps(p)$ can be written as a perturbed Hamiltonian of the form
\begin{equation}\label{resFH}
\mathcal H^\eps(p) =  H_0 + \eps  p \widehat{V}^{\eps,p}(x).
\end{equation}
Here, and in the following, $H_0$ denotes the $\eps$-independent part of the rescaled Hamiltonian, i.e.,
\begin{equation}\label{H}
H_0 := -\Delta_x + |x|^{2\alpha}\Theta(\vartheta)^2,\quad x\in \R^2,
\end{equation}
and, for each fixed $p\in \R$, we have an {\it effective potential} given by
\begin{equation}\label{effpot}
\widehat{V}^{\eps,p}(x) := \eps p + 2|x|^\alpha \Theta(\vartheta).
\end{equation}
Note that the unperturbed operator $H_0$ describes the {\it magnetic confinement} of quantum particles in the $(x_1,x_2)$-plane, since for all $\alpha>0$, the 
potential
\[
V_0(x)=|x|^{2\alpha}\Theta(\vartheta)^2 \to +\infty \quad \text{as $|x|\to \infty$,}
\] 
in view of Assumption \ref{assump}. In particular, the {\it spectrum} $\sigma(H_0)$ is purely discrete, see Lemma \ref{specInfo}. 


\subsection{Effective dynamics} 


Assume, for simplicity, that $\lambda \in \sigma(H_0)$ is a given, 
non-degenerate eigenvalue of $H_0$ with associated eigenfunction $\chi \in L^2(\R^2)$. Then, for each fixed $p\in \R$ and $\eps\in (0,1]$ sufficiently small, we may, 
in view of \eqref{resFH}, regard $(\eps  p) \ll 1$ as an effective perturbation parameter within the fiber decomposition of $\mathcal H^\eps$. 
Using classical techniques from analytic perturbation theory allows us to derive (for $\varepsilon$ small enough) a convergent series for the 
{\it perturbed eigenvalues} $\lambda^\eps$ in the form
\[
\lambda^\eps = \lambda + \eps  p\lambda_1 + (\eps  p)^2\lambda_2 + \dots 
\]
The coefficients $\lambda_j\in \R$ are thereby computed by an iterative procedure starting from the unperturbed 
eigenvalue $\lambda\in \sigma(H_0)$, see 
Section \ref{sec:RalCoef}. for more details. 
In turn, this yields an effective description of $\psi^\eps(t)$ as $\eps\to 0_+$: 

Indeed, let $p_0>0$ denote some (large) cut-off parameter in the momentum variable $p\in \R$ dual to $z$, and assume 
that $\psi_0$ is of the form
\[
\psi_0(x,z) = {a}(z) \chi(x),
\]
where the modulating amplitude $a\in L^2(\R)$ is such that $\widehat a(p)= 0$ for $|p|>p_0$. We shall prove (cf. Theorem \ref{finRes}) that in this case the solution to \eqref{resIVP} satisfies
\begin{equation}\label{approxsol}
\psi^\eps(t, x,z) =  \phi^\eps(t,x, z)  +\mathcal O(\eps (1+|t|)),
\end{equation}
where the $\mathcal O$-notation should be understood w.r.t. the $L^2$-norm, and
\begin{equation}\label{phiint}
\phi^\eps(t,x, z)= e^{-it \lambda /\eps^2}\chi(x)e^{it\lambda_{2}\Delta_z}a\Big(z- \tfrac{ t\lambda_{1}}{\eps}\Big)
\end{equation}
We observe from \eqref{approxsol} and \eqref{phiint} that the solution $\psi^\eps$ is {\it highly oscillatory in time}, with a constant but singular phase $\propto \eps^{-2}$. 
More interestingly, along the $z$-axis the solution exhibits the dispersive behavior of 
a free particle with {\it effective mass} $M= \lambda_{2}^{-1}$, while also being subject to a {\it strong drift} with 
velocity $v^\eps= \lambda_{1}/\eps$. We shall show in Section 2.2. that $\lambda_1>0$ and given by $$\lambda_{1} = 2\, \big\langle \chi, |\cdot|^\alpha \Theta \chi \big\rangle_{L^2(\R^2)},$$
while the formula for $\lambda_2$ is slightly more involved. 
Such strong drifts along $z$ can can also be observed in numerical simulations of 
the corresponding classical particle dynamics, see \cite{geniet}. 

The compact support condition on $\widehat{a}$ is needed in our approach to justify the use of $\eps =\eps p$ 
as a small parameter. To this end, one should note that the solution to \eqref{resIVP} is to be understood via
\begin{equation*}
\psi^\eps (t, x,z)= \mathcal F^{-1} \big(\mathcal S_t(p) \widehat \psi _0(x, p)\big)(z), 
\end{equation*}
where $\mathcal H^\eps(p)$ is the fiber Hamiltonian defined in \eqref{resFH} and $\mathcal S_t(p)$ the associated Schr\"odinger group acting on $L^2(\R^2)$.
In particular, if $\widehat{\psi }_0$ is compactly supported in $p\in \R$, 
then 
\[
\widehat{\psi^\eps }(t,x,p) = \widehat{\psi^\eps }(t,x,p)\mathbbm{1}_{\{|p|\le p_0\}} \quad \text{for all $t \in \R$,} 
\]
since the characteristic function $\mathbbm{1}_{\{|p|\le p_0\}}$ 
clearly commutes with $\mathcal S_t(p)$.

\subsection{Comparison with existing results} When comparing the effective dynamics \eqref{approxsol} to earlier results in the literature, 
we note that in the case of a {\it purely radial} vector potential $$\mathcal A(x)=(0,0,A(|x|)),$$ Yafaev gave a detailed investigation of 
the spectral properties of $\mathcal H$ and the associated long-time behavior of solutions in \cite{yafaev1, yafaev2}. The fact that $A(|x|)$ is purely radial thereby 
allows for a description of the spectrum $\sigma (\mathcal H(p))$ in terms of spectral band-functions. In turn this allows for a representation of the 
long-time behavior of solutions as $t\to \pm \infty$, which is governed by group velocities that play the same role as $v^\eps= \lambda_{1}/\eps$ does in formula \eqref{phiint}. 
We also note that an analogous study for the case of unitary $B$-fields (induced by vector potentials \eqref{Aex} in general spatial dimensions) has later been done in \cite{geniet}. 

In contrast to all of these works we can allow for an azimuthal dependence of the vector potential described by the function $\Theta$. While we do not study the 
asymptotic regime as $t\to \pm \infty$, we instead consider the time-evolution on large time-scales up to $t\sim \mathcal O(\eps^{-1})$ in the 
large field limit as $\eps \to 0_+$. Our results can therefore be seen 
as complementary to \cite{yafaev1, yafaev2}. 


\subsection{Stability under more general perturbations} 


The effective dynamics \eqref{approxsol} is obtained for the particular class of vector potentials given in \eqref{homA}. It is a natural question whether 
a similar effective description holds true for more general, translation invariant magnetic fields. To answer this question, the main observation is that any $\psi^\eps$ which is 
(approximately) given by \eqref{approxsol} remains strongly localized near the origin of the $(x_1, x_2)$-plane. Thus we expect that only the behavior of $\mathcal A^\eps$ in a small 
neighborhood near this origin plays a significant role. 

In the second part of this work, we shall rigorously show that \eqref{approxsol} indeed remains valid for 
a much more general class of vector potentials, given by
\begin{equation}\label{Aheur}
\mathcal A_{\mathfrak a}^\eps(x) =   \frac{1}{\eps^{\alpha +1}}\big(0,0,(|x|^\alpha \Theta(\vartheta)+\mathfrak a(x))\big).
\end{equation}
Here, $\mathfrak a\in L^\infty_{\rm loc}(\R^2)$ is assumed to be vanishing in a small neighborhood around the origin of the $(x_1,x_2)$-plane and 
polynomially bounded at infinity (see Section \ref{sec:sing} for a precise formulation). In particular, $\mathfrak a$ is {\it not necessarily homogenous }and 
can become arbitrarily large as $|x|\to \infty$. The class of vector potentials \eqref{Aheur} is therefore by no means a small perturbation of the previous case \eqref{homA}.

Nevertheless, the heuristic picture described above can be made mathematically rigorous by using the fact that all eigenfunctions $\chi$ of the unperturbed Hamiltonian $H_0$, defined in \eqref{H}, 
admit a strong, i.e. exponential, decay. This implies that 
the contribution of $\mathcal A_\a^\eps(x)$ for large $|x|$ are strongly suppressed in our perturbative analysis. Indeed, if $P_\lambda$ denotes the 
spectral projection onto a simple eigenvalue $\lambda \in \sigma(H_0)$ and $\H_\a^\eps$ the magnetic Hamiltonian with an additional perturbation given by $\a$, 
we shall prove that (cf. Theorem 3.3):
\begin{equation}\label{newest}
\Big \| \Big(e^{-it \H_\a^\eps/\eps^2} - e^{-it \H^\eps /\eps^2} \Big)P_\lambda \Big \|= \mathcal O\big (\eps  (1+|t|)\big ).
\end{equation}
In particular, this implies that for initial data $$\psi_0 (x,z) = P_\lambda \psi_0 (x, z)= a(x) \chi(z),$$ 
with compactly supported amplitude $\widehat a$, we have, by triangle inequality:
\begin{align*}
 \left\|e^{-it \H_\a^\eps/\eps^2}\psi_0 - \phi^\eps(t) \right\|_{L^2(\R^3)}  \le & \, \left\| \Big (e^{-it \H_\a^\eps/\eps^2} - e^{-it \H^\eps/\eps^2}\Big)P_\lambda   \psi_0\right\|_{L^2(\R^3)} \\
& \, + \left\|e^{-it \H^\eps/\eps^2} \psi_0 - \phi^\eps(t) \right\|_{L^2(\R^3)}  \lesssim \eps  (1+|t|).
\end{align*}
In view of \eqref{approxsol} and \eqref{newest}, both terms in the middle are seen to be of the same order, resulting in the last inequality. Thus, the effective dynamics $\phi^\eps$ given by 
\eqref{phiint} is also an approximate solution in the case with an additional inhomogeneous perturbation $\mathfrak a$.


Let us emphasize, however, that a rigorous mathematical justification of this approach is not straightforward. In particular it {\it cannot} be obtained within the framework of 
regular perturbation theory, since $\mathfrak a$ is allowed to grow polynomially as $|x|\to \infty$ and thus is not an $H_0$-bounded perturbation. 
Indeed, the mere assumption of a polynomial bound on $\a$, combined with the fact that $\a$ vanishes in a vicinity of the origin, 
even allows for $\sigma(\mathcal H_\a^\eps)$
to be {\it purely continuous}, see Remark \ref{conspec}. 
To overcome this obstacle one has to use the more sophisticated framework of asymptotic perturbation theory, cf. \cite{Hu, K, reed} for a general overview. More precisely, we 
shall employ the theory of almost invariant subspaces developed in \cite{nenciu}. The latter relies strongly on exponential decay bounds of the 
eigenfunctions of $H_0$. While these bounds can, 
in principle, be obtained from the general theory of Agmon-Combes-Thomas (see e.g. \cite{Hi}), 
for the readers convenience we shall give an elementary proof of this fact in the Appendix.

\medskip

This paper is now {\it organized} as follows: In Section \ref{sec:regular} we shall give a rigorous proof of the effective dynamics described in \eqref{approxsol}. 
We shall also show how to generalize this result to arbitrary initial data in $L^2(\R^3)$, if one is willing to give up on an explicit $\eps$-dependent approximation error. 
The case of polynomially bounded perturbations of $\mathcal A^\eps$ in the form \eqref{Aheur} is then studied in Section \ref{sec:sing}. In it we shall first 
recall some general aspects of the theory of almost invariant subspaces and then apply these results to our setting. 

\smallskip

{\it Notation}: Throughout this work, we shall write $a\lesssim b$ if there exists a constant $C>0$, independent of $\varepsilon\in (0,1]$ and $t\in \R$, such that $a\le C b$.


\section{The homogeneous case: regular perturbation theory}\label{sec:regular}


Assume that $\mathcal A^\eps$ is given by \eqref{homA}, and recall that, after applying the partial Fourier-transformation to \eqref{resIVP}, 
we are lead to consider the following perturbed initial value problem:
\begin{equation}\label{FinalPDE}
	i\eps^2 \partial_t\widehat{\psi}^\eps = \big( H_0 + \eps  p \widehat{V}^{\eps,p}(x)\big) \widehat{\psi}^\eps, \quad \widehat{\psi}^\eps_{\mid t=0} = \widehat{\psi}_0(\cdot, p),
\end{equation}
where $H_0$ is a confining Hamiltonian of the form
\[
H_0=-\Delta_x + |x|^{2\alpha}\Theta(\vartheta)^2,\quad \alpha >0.
\]

Classical results from spectral theory (see, e.g., \cite{MaSh}), describe the main properties of the unperturbed operator $H_0$:

\begin{lemma}\label{specInfo} Assume that $\Theta$ satisfies Assumption \eqref{assump}. Then, for all $\alpha>0$, the operator $H_0$ is essentially self-adjoint on 
$C_0^\infty(\R^2)$. Moreover, we have:
\begin{itemize}
\item[(i)] $H_0\ge 0$ and $(H_0+1)^{-1}$ is compact.

\item[(ii)] The spectrum $\sigma(H_0)$ is an increasing sequence $(\lambda_n)_{n\in \N} \subset \R_+$ of positive eigenvalues with finite multiplicity
and $\lambda_n\to +\infty$ as $n\to \infty$. In particular, 
\[
\sigma(H_0)=\sigma_{\rm disc}(H_0)\ \text{and $\sigma_{\rm ess}(H_0)=\emptyset$.}
\]

\item[(iii)] The associated eigenfunctions $(\chi_n)_{n \in \N}$, counted with multiplicity, form an orthonormal basis of $L^2(\R^2)$.
\end{itemize}
\end{lemma}

In order to ensure that $\widehat{V}^{\eps,p}(x) = \eps p + 2|x|^\alpha \Theta(\vartheta)$ can indeed be considered as a regular perturbation of $H_0$, we have the following simple lemma:
\begin{lemma}\label{lem:Hbound}
For all $u(\cdot, p) \in \D(H_0) \subset L^2(\R^2),$
	\begin{equation*}\begin{split}\label{Vbound}
		\|\widehat{V}^{\eps,p} u(\cdot, p) \|_{L^2(\R^2)} \le \big(\sqrt{2} +\eps |p|\big) \|u(\cdot, p) \|_{L^2(\R^2)} + \sqrt{2} \|H_0u(\cdot, p) \|_{L^2(\R^2)}.
	\end{split}\end{equation*}
Thus, the potential $\widehat{V}^{\eps,p}$ is $H_0$-bounded.
\end{lemma}
\begin{proof}
First, let $u\in C_0^\infty(\R^2_x)$ and denote $v(x)=|x|^{\alpha}\Theta(\vartheta)$.
Since $-\Delta$ is positive, i.e., $\langle u , (-\Delta) u \rangle_{L^2}  \ge 0$, we have, by Cauchy-Schwarz:
\begin{equation*}
		\|v u\|_{L^2}^2 = \langle u, v^2u\rangle_{L^2} \le \langle u,H_0u\rangle_{L^2} \le \|u\|_{L^2} \|H_0u\|_{L^2} 
		\le \frac{1}{2} \big(\|u\|_{L^2} + \|H_0u\|_{L^2}\big)^2.
\end{equation*}			
Hence by triangle inequality,	
\begin{equation*}
\|\widehat{V}^{\eps,p} u \|_{L^2} \le 	\eps  |p| \|u \|_{L^2} +2 \|v u\|_{L^2}  \le 	\big(\sqrt 2+\eps  |p|\big) \|u \|_{L^2} +\sqrt 2 \|H_0u \|_{L^2}.
\end{equation*}
Since $H_0$ is essentially self-adjoint on $C_0^\infty(\R^2)$, this estimate extends to $u \in \mathcal D(H_0)$.	
\end{proof}
Lemma \ref{lem:Hbound} implies that $\H^\eps(p)$  is an 
analytic family of operators in the sense of Kato, and hence we may apply the following result from regular perturbation theory, see \cite{K, reed}:

\begin{proposition}
Let $p\in \R$ be fixed and $\lambda \in \sigma(H_0)$ be an $m$-degenerate eigenvalue of $H_0$. 
Then, there exists an $\eps_0>0$ such that for $|\eps p| < \eps_0 |p|$, and $k=1, \dots, m$ there exist:
\begin{itemize}
\item[(i)] $\lambda_k(\eps p)$, not necessarily distinct, real functions;
\item[(ii)] $\chi_k(\cdot, \eps p) \in L^2(\R^2)$, normalized via $\| \chi_k(\cdot, \eps p)\|_{L^2}=1$;
\item[(iii)] one-dimensional orthogonal projection $P_k(\eps  p): L^2(\R^2)\to L^2(\R^2)$, with \[P_k(\eps p)P_\ell(\eps p) = \delta_{k, \ell} P_{\ell}(\eps p).\]
\end{itemize}
For $|\eps p| < \eps_0 |p|$, all of these have absolutely converging expansion in powers of $(\eps p)$, and they satisfy, for all $k =1, \dots,  m$: 
\[
\lim_{(\eps p)\to 0} \lambda_k(\eps p) = \lambda,
\]
as well as
\[
\H^\eps(p) \chi_k(\cdot, \eps  p) = \lambda_k(\eps  p)\chi_k(\cdot, \eps  p),
\]
and
\[
P_k(\eps p)\chi_k(\cdot, \eps p) = \chi_k(\cdot, \eps p).
\]
The functions $\{\lambda_k(\eps p)\}_{k=1}^m$ are the only eigenvalues of $\H^\eps(p)$ near $\lambda$. 
In particular, $$P_\lambda = \sum_{k=1}^m P_k(0)$$ yields the orthogonal projection onto the spectral subspace of $H_0$ associated to $\lambda$.

\end{proposition}

\begin{remark}
One observes from \eqref{effpot} that the term $\eps p$ within $\widehat V^{\eps, p}$ just shifts the energy by 
a constant factor. In particular, the eigenfunctions for $\mathcal H^\eps(p)$ are the same as for $\mathcal H^\eps(p)-(\eps p)^2$.
\end{remark}

Now, let $\lambda\in \sigma(H_0)$ be an $m$-degenerate eigenvalue of $H_0$. In what follows, we shall write the asymptotic expansions of the perturbed eigenvalues, as $(\eps p)\to 0$,
in the form
\begin{equation}\label{EigenAnalytic}
\lambda_k(\eps  p) = \lambda + \sum_{j=1}^\infty (\eps  p)^j\lambda_{k,j}, \quad k=1,\dots, m, 
\end{equation}
where $\lambda_{k,j}\in \R$. Analogously, we shall write for the associated spectral projections
\begin{equation}\label{Panalytic}
\begin{split}
P_k(\eps  p)= P_k(0) + \sum_{j=1}^\infty (\eps  p)^j Q_{k, j} ,
\end{split}
\end{equation}
and denote 
\[
\chi_k \equiv \chi_k(\cdot, 0).
\]

We can now state the first main result describing the effective dynamics of solutions to \eqref{resIVP}. 

\begin{theorem}\label{finRes} Let $\Theta$ satisfy Assumption \ref{assump} and let the initial data $\psi_0\in L^2(\R^3)$ be such that 
\[
\mathcal F \psi_0 (x,p)\equiv	\widehat{\psi}_0 (x,p)= \widehat{a}(p) \chi_k(x),
\]
where we assume that $\supp \widehat{a}\subset [-p_0, p_0]$, for some $p_0>0$.
Define an approximate solution $\phi^\eps$ to \eqref{resIVP} by
\begin{equation*}\label{phiDef}
	\phi_k^\eps(t, x, z) = e^{-it \lambda /\eps^2}\chi_k(x)e^{it\lambda_{k,2}\Delta_z}a\Big(z-\tfrac{t \lambda_{k,1}}{\eps}\Big).
\end{equation*}
Then, there exists an $\eps_0\in (0,1]$, such that for all $0<\eps<\eps_0$ and all $t\in \R$:
\begin{equation}\begin{split}\label{invFsoln2}
	\left\|e^{-it \H^\eps/\eps^2}\psi_0 - \phi_k^\eps(t) \right\|_{L^2(\R^3)} \lesssim \eps  (1+|t|).
\end{split}\end{equation}
\end{theorem}

\begin{proof}
First, let $p_0>0$ be a momentum cut-off, such that $\widehat{a}(p)=0$ for all $|p|>p_0$. Then, there exists an $\eps_0 > 0$, such that 
both $\lambda_k(\eps  p)$ and $\chi_k(\cdot, \eps  p)\in L^2(\R^2)$ are analytic for
all $\eps < \eps_0$ and $|p| < p_0$ on
\[
I_0=\big(-\eps_0 p_0, \eps_0 p_0\big)\subset \R.
\] 
This then implies that there exists a constant $c_0=c_0(p_0)>0$, such that
\begin{equation}\label{chi}
	\big\|\chi_k - \chi_k(\cdot, \eps  p)\big\|_{L^2(\R^2)} \le c_0 \eps  |p|.
\end{equation}
Keeping in mind that $\widehat{\psi}_0 (x,p)= \widehat{a}(p) \chi_k(x),$ we can thus estimate 
\begin{equation}\label{initData}
	\big\|\widehat{\psi}_0 (\cdot,p) - \widehat{a}(p)\chi_k(\cdot, \eps  p)\big\|_{L^2(\R^2)} \le c_0 \eps  |p|\,|\widehat{a}(p)|.
\end{equation}
Next, we observe that
\begin{equation}\begin{split}\label{soln1}
	e^{-it\eps^{-2}(H_0+\eps  p \widehat{V}^{\eps,p})}\chi_k(x,\eps  p) = e^{-it\eps^{-2}\lambda_k(\eps  p)}\chi_k(x,\eps  p).
\end{split}\end{equation}
In view of \eqref{EigenAnalytic} and using Taylor expansion of the exponential function, we find
\begin{equation*}\begin{split}
	\Big|e^{-it\eps^{-2}\lambda_k(\eps  p)} - e^{-it\eps^{-2}\lambda}e^{-it\eps^{-1}p\lambda_{k,1}}e^{-itp^2\lambda_{k,2}}\Big| \lesssim \eps |p|^3|t|, 
\end{split}\end{equation*}
uniformly on $I_0$.
Combining this expansion with \eqref{chi}, and using a triangle inequality, we obtain 
\begin{equation}\begin{split}\label{evoExp}
	&\Big\|e^{-it\eps^{-2}\lambda_k(\eps  p)}\chi_k(\cdot,\eps  p)
	-e^{-it\eps^{-2}\lambda}e^{-it\eps^{-1}p\lambda_{k,1}}e^{-itp^2\lambda_{k,2}}\chi_k \Big\|_{L^2(\R^2)}\\
	&\le \Big\|e^{-it\eps^{-2}\lambda_k(\eps  p)}(\chi(\cdot,\eps  p)-\chi_k)\Big\|_{L^2(\R^2)} +\\
	&\quad \, + \Big\| \big(e^{-it\eps^{-2}\lambda_k(\eps  p)} - e^{-it\eps^{-2}\lambda}e^{-it\eps^{-1}p\lambda_{k,1}}e^{-itp^2\lambda_{k,2}} \big) \chi_k \Big\|_{L^2(\R^2)}\\
	&\lesssim  \eps  |p|(1+p^2|t|).
\end{split}\end{equation}
Recall that the exact solution to \eqref{FinalPDE} is given by
\[
\widehat{\psi}^\eps (t,x,p)= e^{-it\eps^{-2}(H_0+\eps  p \widehat{V}^{\eps,p})}\widehat{\psi}_0(x,p),
\] 
where $\widehat{\psi}_0 (x,p)= \widehat{a}(p) \chi_k(x)$. Denoting 
\[
	\widehat{\phi}_k^\eps(t,x,p) = e^{-it\eps^{-2}\lambda}e^{-it\eps^{-1}p\lambda_{k,1}}e^{-itp^2\lambda_{k,2}}\widehat{a}(p)\chi_k(x),
\]
we obtain, in view of \eqref{initData} and \eqref{evoExp}, that
\[
\big\|\widehat{\psi}^\eps(t,\cdot,p) - \widehat{\phi}_k^\eps(t,\cdot,p)\big\|_{L^2(\R^2)} \lesssim \eps  |p|(1+p^2|t|)|\widehat{a}(p)|.
\]
This holds uniformly for $p \in I_0$, and so by the compact support of $\widehat{a}(p)$ we may further estimate
\begin{equation}\label{thmEst1}
\big\|\widehat{\psi}^\eps(t) - \widehat{\phi}_k^\eps(t)\big\|_{L^2(\R^3)} \le C \eps  (1+|t|) , 
\end{equation}
for some constant $C=C(p_0, a)>0$.
Next, we consider the approximate solution under the (partial) inverse Fourier transform,
\begin{equation}\begin{split}\label{invFsoln}
	\phi^\eps_k(t,x,z) &= (\mathcal F^{-1}\widehat{\phi}_k^\eps)(t,x,z)\\
	&= e^{-it\eps^{-2}\lambda}\chi_k(x)\frac{1}{\sqrt{2\pi}}\int_\R e^{ip(z-t\eps^{-1}\lambda_{k,1})}\big(e^{-itp^2\lambda_{k,2}}\widehat{a}(p)\big) dp,
\end{split}\end{equation}
and note that the term $e^{-itp^2\lambda_{k,2}}$ is just the Fourier transform of the free evolution $e^{it\lambda_{k,2} \Delta_z}$ in 
$z$-direction. 
Thus, we explicitly obtain
\[
\frac{1}{\sqrt{2\pi}}\int_\R e^{ip(z-t\eps^{-1}\lambda_{k,1})}\big(e^{-itp^2\lambda_{k,2}}\widehat{a}(p)\big)dp = e^{it\lambda_{k,2}\Delta_z}a\big(z-t\eps^{-1}\lambda_{k,1}\big),
\]
where $a$ is the Fourier-inverse of $\widehat{a}$. Finally, applying this to \eqref{thmEst1}, along with Plancherel's theorem, gives us
\begin{equation}
\big\|\psi^\eps(t) - \phi_k^\eps(t)\big\|_{L^2(\R^3)} =\big\|\mathcal F^{-1}\big(\widehat{\psi}^\eps(t) - \widehat{\phi}_k^\eps(t)\big)\big\|_{L^2(\R^3)}
\lesssim  \eps  (1+|t|),
\end{equation}
for all $\eps < \eps_0$. This completes the proof.
\end{proof}


\subsection{More general initial data} \label{sec:general}


The previous result requires us to impose rather severe restrictions on the initial data. In this subsection we shall show 
how to generalize Theorem \ref{finRes} to arbitrary initial data in $L^2$. The price one pays, however, 
is the loss of an explicit, purely $\eps$-dependent convergence rate: 

Indeed, by Lemma \ref{specInfo} we know that $\sigma(H_0) $ is an increasing sequence of eigenvalues $\lambda^n\in \R$, $n\in \N$, with finite-multiplicity $m_n\in \N$. 
Note that here we use the superscript $n$ to index the eigenvalues, not to indicate a power. 
For each $n\in \N$ we denote the perturbed eigenvalues of $\H^\eps(p)$ by $\lambda^n_1(\eps  p), ..., \lambda^n_{m_n}(\eps  p)$, 
where $m_n\in \N$ is the multiplicity of a given $\lambda^n\in \sigma(H_0)$. Similarly to \eqref{EigenAnalytic}, we can express $\lambda^n_k(\eps  p)$ as a power series:
\begin{equation*}
\lambda^n_k(\eps  p) = \lambda^n + \sum_{j=1}^\infty (\eps  p)^j\lambda^n_{k,j}, \quad n\in \N.
\end{equation*}
We recall that the set of unperturbed eigenfunctions associated to $\lambda^n$, i.e.
\[
\big\{\chi^n_k \, : \, 1\le n < \infty, 1\le k \le m_n\big\}
\] 
forms an orthonormal basis of $L^2(\R^2)$. Hence, any $\psi_0 \in L^2(\R^3)\simeq L^2(\R^2)\otimes L^2(\R)$ can be written as
\[
	\psi_0 (x,z)= \sum_{n=1}^\infty\sum_{k=1}^{m_n}a_k^n(z)\chi_k^n(x),
\]
with amplitudes $a_k^n \in L^2(\R_z)$. Under the partial Fourier transform this becomes 
\begin{equation}\label{bDef1}
	\widehat{\psi}_0 (x,p)= \sum_{n=1}^\infty\sum_{k=1}^{m_n}\widehat{a}_k^n(p)\chi_k^n(x),
\end{equation}
where the coefficients $\widehat{a}_k^n \in L^2(\R)$ are not necessarily compactly supported. For such general initial data, we have the following result:

\begin{proposition}\label{infRes} 
Assume $\Theta$ satisfies assumption \eqref{assump} and let $\psi_0 \in L^2(\R^3)$.
Then, for any $\delta > 0$ there exists $N_\delta \in \N$ and $\eps_\delta > 0$ such that for all $0 < \eps < \eps_\delta$, an approximate solution to \eqref{resIVP} is given by
\begin{equation}\label{fullApp}
	\phi^\eps(t, x, z) = \sum_{n=1}^{N_\delta}\sum_{k=1}^{m_n}e^{-it \lambda^n/\eps^2}\chi^n_k(x)e^{it\lambda^n_{k,2}\Delta_z}b_k^n\Big(z-\tfrac{t \lambda^n_{k,1}}{\eps}\Big),
\end{equation}
where all $\widehat b_k^n\in C_{0}(\R)$ are compactly supported $L^2$-approximations of the coefficients $\widehat{a}_k^n$ appearing in \eqref{bDef1}, 
and we have the estimate
\begin{equation}\begin{split}
\left\|e^{-it \H^\eps/\eps^2}\psi_0 - \phi^\eps(t)\right\|_{L^2(\R^3)}  \lesssim \delta(1+|t|),
\end{split}
\end{equation}
\end{proposition}

\begin{proof} By linearity of \eqref{resIVP}, we will be able to reduce our initial data to a finite sum and apply the same analysis from Theorem \ref{finRes}: 
First, for any $\delta > 0$ we can choose $N_\delta \in \N$ large enough, such that
\[
	\widehat{\varphi}_1 (x,p) = \sum_{n=1}^{N_\delta}\sum_{k=1}^{m_n}\widehat{a}_k^n(p)\chi_k^n(x)
\]
satisfies
\begin{equation*}
	\big\|\widehat{\psi}_0 - \widehat{\varphi}_1\big\|_{L^2(\R^3)} \le \delta.
\end{equation*}
Then, for any $\delta_1 > 0$ we may approximate each $\widehat{a}_k^n\in L^2(\R)$, where $1\le n \le N_\delta$, $1\le k \le m_n$, by a compactly supported 
$\widehat{b}_k^n \in C_0(\R)$, such that
\begin{equation}\label{bDef2}
	\big\|\widehat{a}_k^n - \widehat{b}_k^n\big\|_{L^2(\R)} \le \delta_1.
\end{equation}
Defining
\[
	\widehat{\varphi}_2 (x,p) = \sum_{n=1}^{N_\delta}\sum_{k=1}^{m_n}\widehat{b}_k^n(p)\chi_k^n(x),
\]
we thus have
\[
	\big\|\widehat{\psi}_0 - \widehat{\varphi}_2\big\|_{L^2(\R^3)} \lesssim \sum_{n=1}^{N_\delta}m_n \delta_1+ \delta.
\]
Finally, let
\begin{equation}\label{fullAppInit}
	\widehat{\varphi}_3 (x,p) = \sum_{n=1}^{N_\delta}\sum_{k=1}^{m_n}\widehat{b}_k^n(p)\chi_k^n(x,\eps  p),
\end{equation}
then, by \eqref{chi} we can estimate
\begin{equation*}\begin{split}
	\big\|\widehat{\psi}_0 - \widehat{\varphi}_3\big\|_{L^2(\R^3)} \lesssim \sum_{n=1}^{N_\delta}m_n(\eps +\delta_1) + \delta.
\end{split}\end{equation*}

Having reduced our initial data to a finite sum of terms compactly supported in $p\in \R$, we may now directly apply Theorem \ref{finRes} to each term in the approximate initial data given by \eqref{fullAppInit}. Hence, we obtain
\begin{equation*}\begin{split}
	\big\|\psi^\eps (t) - \phi^\eps(t)\big\|_{L^2(\R^3)}\lesssim\sum_{n=1}^{N_\delta}m_n\big(\eps (1+|t|) + \delta_1 \big) + \delta,
\end{split}\end{equation*}
where $\phi^\eps$ is given by \eqref{fullApp} and $b_k^n$ denotes the Fourier-inverse of $\widehat b_k^n$.
To finally bound the error terms, we first choose $\delta > 0$ arbitrarily small. 
This choice will fix the values of $N_\delta$ and $\sum_{n=1}^{N_\delta}m_n$. We can then choose $\eps_\delta<1$ and $\delta_1<1$ small enough, 
such that for all $0 < \eps < \eps_\delta$,
\[
	\sum_{n=1}^{N_\delta}m_n\big(\eps (1+|t|) + \delta_1 \big) \lesssim \delta(1+|t|).
\]
This yields the claim.
\end{proof}


\subsection{Computation of the main perturbation coefficients} \label{sec:RalCoef} 

We finally turn to the computation of the perturbation coefficients appearing in our leading order approximation $\phi^\eps$. 
Recall that for the $n$-th eigenspace of $H_0$, the 
corresponding perturbed eigenvalue of $\H^\eps(p)$ can be written as
\begin{equation}\label{EigenAnalytic2}
\lambda^n_k(\eps  p) = \lambda^n + \eps  p \lambda^n_{k,1} + (\eps  p)^2 \lambda^n_{k,2}+\mathcal O((\eps  p)^3).
\end{equation}
Since only $\lambda^n_{k,1}$ and $\lambda^n_{k,2}$, enter into the definition of $\phi^\eps$, we shall in the following only
focus on the computation of these two coefficients.

Indeed, using well-known formulas from perturbation theory (see, e.g., \cite{reed}), we have:
\begin{equation*}\begin{split}
	\widetilde{\lambda}_{k,1}^n &= \big \langle \chi_k^n, \widehat{V}^{\eps,p}\chi_k^n \big \rangle_{L^2(\R^2)},\\
	\widetilde{\lambda}_{k,2}^n &= -\sum_{\ell =1, \ell \neq n}^\infty\sum_{j=1}^{m_\ell} ( \lambda^\ell - \lambda^n)^{-1} 
	\big|\big \langle \chi_k^n, \widehat{V}^{\eps,p}\chi_j^\ell \big\rangle_{L^2(\R^2)}\big|^2.
\end{split}\end{equation*}
Here, we use the tilde notation to indicate that these are not yet the final values for $\lambda^n_{k,1}$ and $\lambda^n_{k,2}$, 
since $\widehat{V}^{\eps,p}$ will introduce some additional factors of $\eps  p$. 
To obtain an expansion of the form \eqref{EigenAnalytic2}, we thus need to recombine terms according to their respective order in $\eps  p$. To this end, 
we calculate
\begin{equation*}\begin{split}
	\big \langle \chi_k^n(\cdot,0), \widehat{V}^{\eps,p}\chi_j^\ell  \big \rangle_{L^2(\R^2)}& = 
	\big\langle \chi_k^n , \big(\eps p + 2|\cdot|^\alpha \Theta \big)\chi_j^\ell \big\rangle_{L^2(\R^2)}\\
	&= 2 \, \big\langle \chi_k^n, |\cdot|^\alpha \Theta \chi_j^\ell \big\rangle_{L^2(\R^2)} + \eps p\, \delta_{k,j}\delta_{n,\ell},
\end{split}\end{equation*}
where $\delta_{k,j}$ is the Kronecker delta. Denoting 
\[
	v_{k,j}^{n,\ell} := 2\, \big\langle \chi_k^n, |\cdot|^\alpha \Theta \chi_j^\ell \big\rangle_{L^2(\R^2)}\in \R_+,
\]
we find $\widetilde{\lambda}_{k,1}^n = v_{k,k}^{n,n} + \eps p$, and 
\begin{equation*}\begin{split}
\widetilde{\lambda}_{k,2}^n &= -\sum_{\ell=1, \ell  \neq n}^\infty\sum_{j=1}^{m_\ell} (\lambda^\ell - \lambda^n)^{-1}\big|v_{k,j}^{n,\ell} +\eps p\delta_{k,j}\delta_{n,\ell}\big|^2\\
	&= -\sum_{\ell=1, \ell \neq n}^\infty\sum_{j=1}^{m_\ell} (\lambda^\ell - \lambda^n)^{-1}|v_{k,j}^{n,\ell} |^2.
\end{split}\end{equation*}
Hence, in accordance with the power series \eqref{EigenAnalytic2} we have:
\begin{equation*}
	{\lambda}_{k,1}^n = v_{k,k}^{n,n},\quad
	{\lambda}_{k,2}^n = 1 - \sum_{l=1, \ell \neq n}^\infty\sum_{j=1}^{m_\ell} (\lambda^\ell - \lambda^n)^{-1}|v_{k,j}^{n,\ell} |^2.
\end{equation*}
Unfortunately, it seems that for general $\alpha>0$ there is no explicit expression available for the $L^2$-inner products which define the coefficients $v_{k,j}^{n,\ell}$, even in the case where $\Theta \equiv 1$.


\section{Stability under perturbations vanishing in the vicinity of the origin}\label{sec:sing}


In this section we show that the previous results are stable against a large class of singular perturbations. To this end, we consider magnetic vector potentials of the form:
\begin{equation}\label{Aalphab}
\mathcal A_\a^\eps(x)=\frac{1}{\eps^{\alpha+1}} \big(0, 0, |x|^\alpha \Theta(\vartheta) +\a(x)\big), \quad \alpha>0,
\end{equation}
where $\a\in L^\infty_{\rm loc}(\R^2)$ is polynomially bounded at infinity and vanishes in an arbitrarily small neighborhood of the origin. 
More specifically, we impose:
\begin{assumption}\label{Aassump}
There exist constants $x_0>0$, $M<\infty$, and $\beta<\infty$, such that 
\begin{equation*}
\a(x) = \mathbbm{1}_{\{|x|\ge x_0\}} \a(x) \quad \text{and} \quad  |\a(x)| \le M |x|^\beta.
\end{equation*}
\end{assumption}
Note that, since for any $s>0$:
\[
\sup_{|x|\ge x_0} |x|^{-s} = \frac{1}{x_0^s},
\]
 Assumption \ref{Aassump} implies that 
\[
|\a(x)| \le \frac{M}{x_0^s} |x|^{\beta+s}.
\]
Thus, one can take w.l.o.g. $\beta>0$ sufficiently large, to ensure
\begin{equation}\label{alphadelta}
 \beta\ge 3+\alpha.
\end{equation}

Using the same rescaling of the spatial variables as in Section 1, a straightforward computation shows that instead of \eqref{resFH}, one obtains
\begin{equation}\label{resFHsing}
\begin{split}
\mathcal H_\a^\eps(p) = &\, H_0 + \eps  p \widehat{V}^{\eps,p}(x) 
+\eps^{\beta-\alpha }\widehat{W}_\a^{\eps,p}(x)\\
= & \, \mathcal H^\eps(p) +\eps^{\beta-\alpha}\widehat{W}_\a^{\eps,p}(x),
\end{split}
\end{equation}
where $H_0$ and $\widehat{V}^{\eps,p}$ are the same as in \eqref{H} and \eqref{effpot}, respectively, and
\begin{equation}\label{W}
\begin{split}
\widehat{W}_\a^{\eps,p}(x) =   2p \eps^{1-\beta}\a(\eps x) + 2\eps^{-\beta }|x|^\alpha \Theta(\vartheta)\a(\eps x)   + \eps^{-\beta -\alpha }(\a(\eps x))^2.
\end{split}
\end{equation}
\begin{lemma} Assumption \ref{Aassump} implies that $\widehat W^\eps_{\mathfrak a}$ is polynomially bounded, uniformly for $0<\eps  \le 1$, i.e.
\begin{equation}\label{bddW}
\big| \widehat{W}_\a^{\eps}(x)\big| \lesssim \big(1+|x|\big)^{2\beta}.
\end{equation}
\end{lemma}
\begin{proof}
We first estimate, using  Assumption \ref{Aassump} and the fact that $\Theta $ is bounded:
\begin{align*}
|\widehat{W}_\a^{\eps,p}(x)| \lesssim  &\,  |p| \eps^{1-\beta}|\a(\eps x)| +  \eps^{-\beta } |x|^\alpha |\a(\eps x)| + \eps^{-\beta -\alpha } |\a(\eps x)|^2  \\
\lesssim  & \, (1+ |p|) \big (\eps |x|^\beta +  |x|^{\alpha +\beta} + \eps^{\beta -\alpha} |x|^{2\beta} \big)
\end{align*}
We now notice that $\beta -\alpha\ge 3$, in view of \eqref{alphadelta}, and thus all the powers of $\eps$ are greater than zero, which yields
\[
|\widehat{W}_\a^{\eps,p}(x)| \lesssim (1+|p|)(1+|x|)^{2\beta}.
\]\end{proof}

This setting will allow us to prove that the result of Section \ref{sec:regular} remains valid. More precisely we shall show:

\begin{theorem}\label{finResing}
Suppose $\mathcal A_\a^\eps$ is given by \eqref{Aalphab}, where $\a$ satisfies Assumption \ref{Aassump}. Let $\lambda \in \sigma(H_0)$ and $P_\lambda$ the corresponding spectral projection. 
Then there exists $\eps_0\in (0,1]$, such that for all 
$0<\eps<\eps_0$:
\begin{equation*}
\left\|\Big(e^{-it \H_\a^\eps/\eps^2} - e^{-it \H^\eps/\eps^2} \Big)P_\lambda \right\|  \lesssim \eps  (1+|t|).
\end{equation*}
Here, $\H_\a^\eps$ and $\H^\eps$ are defined via their fiber representations given in \eqref{resFHsing}.
\end{theorem}

For the proof of this theorem, we shall rely on the asymptotic perturbation theory developed in \cite{nenciu} (see also \cite{Hu, K} and the 
references therein). Sending the reader to \cite{nenciu} for more technical details, we shall give in the next subsection the main ideas behind this perturbation theory and the 
construction of almost invariant subspaces. 
Using these, the proof of Theorem \ref{finResing} will then be given in Subsection \ref{ssec:proof}, where we end with some concluding remarks.


\subsection{Construction of almost invariant subspaces}\label{ssec:inv} 


The fact that $\a$ is in general not homogenous implies that, in contrast to Section 2, the perturbation parameter is no longer given by $\eps p$, but instead $\eps$ and $p$ have 
to be seen as independent from each other. However, it is clear from the previous discussions that, once a momentum cut-off $p_0>0$ has been chosen, 
the precise value of $|p|\le p_0$ no longer matters. We shall therefore stop 
tracking the dependence on $p$ and instead only focus on the dependence on $\eps \ll 1$. To emphasize this fact, we introduce the following shorthand notation: Let 
\begin{equation*}\label{not}
\mathcal V_\a^\eps  (x) := p \widehat{V}^{\eps,p}(x) 
+\eps ^{\beta-\alpha -1}\widehat{W}_\a^{\eps,p}(x),
\end{equation*}
so that for $\a \equiv 0$ we get back $\mathcal V_0^\eps   \equiv  p \widehat{V}^{\eps,p}$. Using this, we consider the perturbed Hamiltonian $\mathcal H^\eps _\a$ in the form
\[
\mathcal H^\eps _\a   = H_0+ \eps  \mathcal V_\a^\eps (x).
\]
From \eqref{bddW} and the explicit form of $\widehat{V}^{\eps,p}(x) $, we know that $\mathcal V_\a^\eps  (x)$ is polynomially bounded uniformly
for $0<\eps  \le 1$:
\begin{equation}\label{bddpV}
\big| \mathcal V_\a^\eps (x)\big| \lesssim (1+p_0^2)\big(1+|x|\big)^{2\beta}.
\end{equation}

Now, let $\lambda \in \sigma(H_0)$, $P_\lambda$ be the corresponding spectral projection, and recall that if
$\a \equiv 0$, the potential $\mathcal V_0^\eps $ is $H_0$-bounded (see Section 2). Thus, for $\eps \in (0,1]$ sufficiently small, we already know from regular perturbation theory that 
\begin{equation}\label{Petabdd}
P_{\lambda}(\eps )=  \sum_{j=0}^{\infty}\eps ^jQ_j(\eps ) \equiv P_\lambda + \sum_{j=1}^{\infty}\eps ^jQ_j(\eps ).
\end{equation}
The $Q_j$ can thereby be expressed via (cf. \cite{K}):
\begin{equation}\label{Ej}
Q_j(\eps ) =(-1)^j\frac{i}{2\pi}\oint_\Gamma (H_0-z)^{-1}\big(\mathcal V_0^\eps  (H_0-z)^{-1}\big)^j dz,
\end{equation}
where
$$
\Gamma =\Big\{ z\in \mathbb C \, : \, |z-\lambda |= \tfrac{d}{2} \Big\},\quad d={\rm dist}\, (\lambda, \sigma(H_0)\setminus\{\lambda\}).
$$
From \eqref{Petabdd} and $P_{\lambda}=P_{\lambda}^2$, we can compute
\[
	\sum_{j=0}^{\infty}\eps^jQ_j(\eps) = \Big(\sum_{i=0}^{\infty}\eps^iQ_i(\eps)\Big)\Big(\sum_{k=0}^{\infty}\eps^kQ_k(\eps)\Big),
\]
and matching terms with equal powers of $\eps$, we find
\[
	Q_j(\eps)=\sum_{l=0}^jQ_l(\eps)Q_{j-l}(\eps).
\]
Similarly, one can use \eqref{Petabdd} to expand $P_{\lambda}(\eps)\mathcal H^\eps  = \mathcal H^\eps P_{\lambda}( \eps)$, and combining terms by their power in $\eps$, one finds
\begin{equation*}\label{EjV}
[H_0, Q_{j+1}(\eps)]=-[\mathcal V_0^\eps, Q_{j}(\eps)].
\end{equation*}

In the case $\a \not =0$, however, we cannot directly follow the same approach, since \eqref{bddpV} does not imply that $\mathcal V_\a^\eps $ is $H_0$-bounded. 
Thus, \eqref{Ej} is {\it a-priori} not well defined if $\mathcal V_0^\eps $ is replaced by $\mathcal V_\a^\eps $. 

\begin{remark}\label{conspec}
In the particular case where $\a\in C^1(\R^2)$ and such that
\[
\a(x) = \begin{cases}
\quad 0  \quad \quad \quad \, \quad \text{ for $|x|\le 1$},\\
 - |x|^\alpha  \Theta(\vartheta)  \quad \text{ for $|x|\ge 2$}, \end{cases}
\]
we have
\begin{equation*}
\mathcal V^\eps_\a(x) = 
\begin{cases} \mathcal V^\eps_0(x) \quad \ \text{for $|x| \le \tfrac{1}{\eps}$}, \\
\eps p^2 \quad \quad \, \,  \text{for $|x|\ge  \tfrac{2}{\eps}$.} \end{cases}
\end{equation*}
This implies that the spectrum of $\mathcal H^\eps_\a$ includes the interval $[(\eps p)^2, \infty)\subset \R$ and thus $\sigma(\mathcal H^\eps_\a)$ is no longer discrete, 
in contrast to $\sigma(\mathcal H^\eps_0)$.
\end{remark}

The way out of this problem 
is to use  an alternative formula for $Q_j(\eps )$, obtained from \eqref{Ej} via the residue theorem (see e.g. \cite{reed}, or \cite[Chapter II.2.1]{K}  in the self-adjoint case):
\begin{equation}\label{altQ}
Q_j(\eps )=(-1)^{j+1}\sum_{\nu_1,...,\nu_{j+1} \geq 0; \Sigma_{l=1}^{j+1}\nu_l=j}
S^{\nu_1}\mathcal V_0^\eps  S^{\nu_2}\mathcal V_0^\eps ...\mathcal V_0^\eps  S^{\nu_{j+1}}.
\end{equation}
Here, $S$ is the reduced resolvent of $H_0$ at $\lambda$, i.e.
\begin{equation}\label{S}
S=\frac{i}{2\pi}\oint_\Gamma (H_0-\zeta)^{-1} (\lambda -\zeta)^{-1} \, d\zeta,
\end{equation}
and, by convention, for $\nu_j=0$:
\begin{equation}\label{RedResCon}
S^0=-P_\lambda.
\end{equation}

The fact that \eqref{altQ} remains well-defined if $\mathcal V_0^\eps $ is replaced by $\mathcal V_\a^\eps $ is the main ingredient in the development of an asymptotic perturbation theory. Indeed, we have 
the following lemma, originally proved in \cite{nenciu}:
\begin{lemma}\label{Ejeta} For $j \in \N$, let 
\begin{equation}\label{EjS}
Q^\a_j(\eps )=(-1)^{j+1}\sum_{\nu_1,...,\nu_{j+1} \geq 0; \Sigma_{l=1}^{j+1}\nu_l=j}
S^{\nu_1}\mathcal V_\a^\eps  S^{\nu_2}\mathcal V_\a^\eps ...\mathcal V_\a^\eps  S^{\nu_{j+1}}.
\end{equation}
Then, all $Q^\a_j(\eps )$ are well defined and bounded on $C_0^\infty (\mathbb R^2)$, and their 
extensions by continuity satisfy $Q^\a_j(\eps ) L^2(\mathbb R^2) \subset \mathcal D(H_0) \cap \mathcal D(\mathcal V_\a^\eps )$, as well as
\begin{equation}\label{Ejeta2}
Q^\a_j(\eps ) =\sum_{l=0}^j Q^\a_l(\eps ) Q^\a_{j-l}(\eps ).  
\end{equation}
Denoting by $[\cdot, \cdot]$ the usual commutator bracket, we also have
\begin{equation}\label{Ejeta3}
[H_0, Q^\a_j(\eps )] f =-[\mathcal V_\a^\eps , Q^\a_{j-1}(\eps )] f,\ \text{for $f \in \mathcal D(H_0) \cap \mathcal D(\mathcal V_\a^\eps )$},
\end{equation}
and
\begin{equation}\label{Ejeta4}
[\mathcal V_\a^\eps , Q^\a_j(\eps )] \ \text{is bounded on $\mathcal D(\mathcal V_\a^\eps )$, for all $j\in \N$}.
\end{equation}
\end{lemma}

The main point of this lemma concerns the boundedness of $Q^\a_j(\eps )$, which in itself relies on the fact that the reduced 
resolvent $S$ preserves the exponential decay of the eigenfunctions of $H_0$, see the appendix. 

\begin{proof}[Proof (Sketch)] 
First, by taking $K=\Gamma$ in Proposition \ref{ResProp} and applying it to \eqref{S}, there exists $\omega_\Gamma >0$ and $M_\Gamma <\infty$ such that 
for all $\omega \in [0, \omega_\Gamma]$:
\begin{equation}\label{bddS}
\Big \|e^{\omega \langle\, \cdot\, \rangle}Se^{-\omega\langle\, \cdot\, \rangle}\Big \| \le M_\Gamma,
\end{equation}
where $\langle x \rangle = (1+|x|^2)^{1/2}$. Similary, as a consequence of Proposition \ref{EigenProp}, there exists $M_{\omega_\Gamma} < \infty$ such that
\begin{equation}\label{expPEst}
	\Big \| e^{\omega_\Gamma\langle\, \cdot\, \rangle}P_\lambda e^{\omega_\Gamma\langle\, \cdot\, \rangle}\Big \| \le M_{\omega_\Gamma}.
\end{equation}
Notice that each term in the sum of \eqref{EjS} contains at least one $S^0=-P_\lambda$ by \eqref{RedResCon} and the restrictions on $\nu_l$. 
In particular, 
\[ 
Q_0(\eps )=Q_0^\a(\eps)=P_\lambda.
\]
We may therefore insert appropriate $e^{\omega \langle\, \cdot\, \rangle}e^{-\omega \langle\, \cdot\, \rangle}$ (or its opposite) between terms in \eqref{EjS} in order to group the appearing factors
$\mathcal V_\a^\eps $ (which are polynomially bounded in view of \eqref{bddpV}) with exponential decay terms $e^{-\omega \langle\, \cdot\, \rangle}$, 
while grouping the growth terms $e^{\omega_\Gamma\langle\, \cdot\, \rangle}$ with $S$ and $P_\lambda$ in such a way that these 
terms are bounded by \eqref{bddS} and \eqref{expPEst}. Hence, we obtain that \eqref{EjS} is 
indeed well defined and bounded. The algebraic properties \eqref{Ejeta2}--\eqref{Ejeta4} can be derived along the same lines as in the case of regular perturbation theory.
\end{proof}

With Lemma \ref{Ejeta} in hand, the construction of almost invariant subspaces of $H^\eps _\a$ follows closely Section III of \cite{nenciu}: 
to begin with, we consider the truncated perturbation series
\begin{equation}\label{TN}
T^\a_N(\eps )=\sum_{j=0}^N\eps ^jQ^\a_j(\eps ),
\end{equation}
where $Q^\a_j$ is given by \eqref{EjS}. 
The operator $T^\a_N$ defines an almost projection, since from \eqref{TN} and \eqref{Ejeta2}, it follows
\begin{equation}\label{TNEst}
\Vert T^\a_N(\eps )^2 -T^\a_N(\eps ) \Vert \lesssim  \eps ^{N+1}.
\end{equation}
Since $Q_0^\a(\eps)=P_\lambda$, we also have
\[
\lim_{\eps  \rightarrow 0}\Vert T^\a_N(\eps )- P_{\lambda}\Vert =0.
\]
The next step is to construct a {\it bona fide} orthogonal projection close to $T_N(\eps)$. To this end, we recall the following abstract result on almost projections:
\begin{proposition}\label{TP}
Let $T$ be a bounded self-adjoint operator satisfying 
\begin{equation*}
\Vert T^2-T \Vert \le c <\frac{1}{8}.
\end{equation*}
Then, $\big\{z\in\C\, : \, |z-1|=\frac{1}{2}\big\} \subset \rho(T)$ and
\begin{equation*}\label{PTles}
\Vert P-T\Vert \lesssim c,
\end{equation*}
where
\begin{equation*}\label{P}
P= \frac{i}{2\pi}\oint_{|z-1|= \tfrac{1}{2}} \, (T-z)^{-1} \, dz.
\end{equation*}
\end{proposition}
\begin{proof}
See Section 3 of \cite{nenciu}.
\end{proof}

In view of \eqref{TNEst}, $T^\a_N(\eps )$ satisfies the condition for this proposition, provided $\eps $ is small enough. We may then define the orthogonal projection
\begin{equation}\label{PN}
P^\a_N(\eps ):= \frac{i}{2\pi}\oint_{|z-1|= \tfrac{1}{2}} (T^\a_N(\eps )-z)^{-1} \, dz,
\end{equation}
which, in view of \cite[identity (6.3)]{nenciu}, satisfies for $\eps$ sufficiently small:
\begin{equation}\label{PT}
\begin{split}
& P^\a_N(\eps ) - T^\a_N(\eps ) = \big( T^\a_N(\eps )^2 -T^\a_N(\eps )\big) \\
& \, \times \frac{i}{2\pi} \oint_{|z-1|= \tfrac{1}{2}}  \frac{1}{z(1-z)}  \left(1+\frac{T^\a_N(\eps )}{z(1-z)}\right)  \left(1+\frac{ T^\a_N(\eps )^2 -T^\a_N(\eps )}{z(1-z)}\right)^{-1}  dz.
\end{split}
\end{equation}
Proposition \ref{TP} therefore implies
\begin{equation}\label{PTeta}
\Vert P^\a_N(\eps ) -T^\a_N(\eps ) \Vert \lesssim  \eps ^{N+1}.
\end{equation}
Moreover, \eqref{Ejeta3} implies
\begin{equation*}\label{HTNeta}
[\H^\eps _\a, T^\a_N(\eps )]= \eps ^{N+1}[\mathcal V_\a^\eps , Q^\a_N(\eps )],
\end{equation*}
and hence, from \eqref{PN} we find
\begin{equation*}
[\H^\eps _\a, P^\a_N(\eps )] =-
\eps ^{N+1} \frac{i}{2\pi}\oint_{|z-1|= \tfrac{1}{2}} (T_N(\eps )-z)^{-1}[\mathcal V_\a^\eps , Q^\a_N(\eps )] (T^\a_N(\eps )-z)^{-1}dz.
\end{equation*}
This together with \eqref{Ejeta4} gives the crucial estimate
\begin{equation}\label{aiP}
\big\Vert [\H^\eps _\a, P^\a_N(\eps )] \big\Vert \lesssim \eps ^{N+1},
\end{equation}
i.e., in the language of \cite[Section II]{nenciu}: 
\[
\text{$P^\a_N(\eps )L^2(\mathbb R^2)$ {\it are almost invariant subspaces of order $\eps ^{N+1}$ associated to $\H^\eps _\a$.}}
\] 
For our purposes, $P^\a_N(\eps )$ will serve to replace the perturbed spectral projections in the analytic case. 

\medskip

Next, we recall that $\mathcal V_\a^\eps  = p \widehat{V}^{\eps,p}+\eps ^{\beta-\alpha -1}\widehat{W}_\a^{\eps,p}$. Inserting this expression into \eqref{EjS} and 
collecting all the terms which do not contain $\widehat{W}_\a^{\eps,p}$ allows us to rewrite
\[
Q_j^\a(\eps) = Q_j (\eps) + \eps^{\beta -\alpha -1} R_j^\a(\eps), 
\]
and since $R_j^\a(\eps)$ inherits the boundedness from $Q_j^\a(\eps)$, we have
\begin{equation}\label{Qest}
\| Q_j^\a(\eps) - Q_j (\eps) \| \lesssim \eps^{\beta -\alpha -1} .
\end{equation}
With estimates \eqref{PTeta}, \eqref{aiP}, and \eqref{Qest} at hand, we can now turn to the proof of Theorem \ref{finResing}.


\subsection{Proof of Theorem \ref{finResing}}\label{ssec:proof} 

We start by collecting the following facts, all of which will be used without further notice in the proof: 

Since $P_\lambda(\eps)$ is a finite dimensional projection, 
the results in the appendix applied to $\mathcal H^\eps(p)$ imply that $\widehat{W}_\a^{\eps,p}P_\lambda(\eps)$, and thus 
also $\H^\eps_\a P_\lambda(\eps)$, is bounded. Moreover, the arguments showing boundedness of $Q_j^\a(\eps)$ also yield boundedness of $\H^\eps_\a Q^\a_j(\eps)$, 
which in itself implies that $\H^\eps_\a T^\a_N(\eps)$ is bounded. In turn this yields boundedness of $\H^\eps_\a P^\a_N(\eps)$ by invoking formula \eqref{PT}. 
Furthermore, since $\H^\eps_\a$, $P^\a_N(\eps)$, and $P_\lambda(\eps)$ are all self-adjoint, we have that
$$\| \H^\eps_\a P_\lambda(\eps)\| = \| P_\lambda(\eps) \H^\eps_\a\| \quad \| \H^\eps_\a P^\a_N(\eps) \| = \| P^\a_N(\eps) \H^\eps_\a\|.$$ 

With this in mind, our goal is to derive an asymptotic estimate for the operator norm of
\[
\Delta(\eps,t) := \Big(e^{-i \eps^{-2} t \H_\a^\eps} - e^{-i  \eps^{-2} t \H^\eps} \Big)P_\lambda.
\]

To do so, we start by using the triangle inequality, together with the fact that both of the Schr\"odinger groups are unitary, to obtain
\begin{align*}
\Big \| \Delta(\eps,t) - \Big ( e^{-i \eps^{-2}  t \H_\a^\eps} P_N^\a(\eps)- e^{-i  \eps^{-2} t \H^\eps} P_\lambda (\eps)\Big )  \Big \| \le \| P_N^\a(\eps) - P_\lambda \| + \| P_\lambda(\eps) - P_\lambda \|.
\end{align*}
By \eqref{Petabdd}, we have
\[
\| P_\lambda(\eps) - P_\lambda \| \lesssim \eps,
\]
while \eqref{TN} and \eqref{PTeta} imply
\begin{equation}\label{Pdiffest}
\| P_N^\a(\eps) - P_\lambda \| \lesssim \| P_N^\a(\eps) - T_N^\a (\eps) \| + \| T_N^\a(\eps) - P_\lambda \| \lesssim \eps^{N+1}+\eps\lesssim \eps.
\end{equation}
Thus
\[
\Big \| \Delta(\eps,t) - \Big ( e^{-i \eps^{-2}  t \H_\a^\eps} P_N^\a(\eps)- e^{-i  \eps^{-2} t \H^\eps} P_\lambda (\eps)\Big )  \Big \|  \lesssim  \eps,
\]
and since $[\mathcal H^\eps, P_\lambda(\eps) ] =0$, we can rewrite
\begin{align*}
e^{-it \eps^{-2} \H^\eps} P_\lambda(\eps) = e^{-it \eps^{-2} P_\lambda(\eps) \H^\eps P_\lambda(\eps)} P_\lambda(\eps),
\end{align*}
which yields
\begin{align}\label{estt1}
\Big \| \Delta(\eps,t) - \Big ( e^{-i \eps^{-2} t \H_\a^\eps} P^\a_N(\eps) - e^{-it \eps^{-2} P_\lambda(\eps) \H^\eps P_\lambda(\eps)} P_\lambda(\eps)\Big )  \Big \| \lesssim \eps .
\end{align}

To further estimate the difference of the two Schr\"odinger groups involved, we recall the following general fact: 
Let $A$ and $B$ be two self-adjoint operators, such that $A-B$ is bounded, and denote
\[
W_s : = e^{i s A} e^{-i s B}, \quad \text{for $s\in \R$.}
\]
Then $W_0=1$ and 
\[
W_s = 1 + \int_0^s \dot W_\tau \, d\tau . 
\]
Writing out the time-derivative $\dot W_\tau$ and multiplying by $e^{-i s A}$, yields Duhamel's formula (or, equivalently, Dyson's formula in the interaction picture) for the 
difference of two Schr\"odinger groups, i.e.
\begin{align}\label{duham}
e^{-is A }- e^{-is B}  = -i \int_0^s e^{-i(s-\tau) A} \big( A - B \big ) e^{-i \tau B} e^{-i \tau A}  \,d\tau.
\end{align}

Next, we decompose $\H_\a^\eps $ into its diagonal and off-diagonal elements, i.e.
\begin{align*}
\H_\a^\eps = & \, P^\a_N(\eps) \H_\a^\eps P^\a_N(\eps) + (1-P^\a_N(\eps)) \H_\a^\eps (1-P^\a_N(\eps)) +\\
& \, + (1-P^\a_N(\eps)) \H_\a^\eps P^\a_N(\eps) + P^\a_N(\eps) \H_\a^\eps (1-P^\a_N(\eps)) \\
: = & \,  \H_{\a, \rm diag} ^\eps + \H_{\a, \rm off} ^\eps.
\end{align*}
In view of \eqref{aiP}, we have
\begin{equation*}
\| \H_{\a, \rm off} ^\eps \| = \big \| (1- 2 P^\a_N(\eps)) [\H^\eps _\a, P^\a_N(\eps )] \big\|\lesssim \eps ^{N+1}.
\end{equation*}
Duhamel's formula \eqref{duham} for $s= \frac{t}{\eps^2}$, $A = \H_\a^\eps$, and $B= \H_{\a, \rm diag} ^\eps$, consequently implies
\begin{align*}
\Big \|  \Big ( e^{-i \eps^{-2} t \H_\a^\eps} - e^{-i \eps^{-2} t \H_{\a, \rm diag}^\eps} \Big ) P^\a_N(\eps) \Big \| \le \int_0^{t \eps^{-2} } \| \H_{\a, \rm off} ^\eps \| \, d\tau \lesssim \eps^{N-1} |t|.
\end{align*}

To proceed further, we note that 
\[
\H_{\a, \rm diag}^\eps \Big|_{\text{\rm ran} P^\a_N(\eps)} = P^\a_N(\eps) \H_\a^\eps P^\a_N(\eps) \Big|_{\text{\rm ran} P^\a_N(\eps)} .
\]
The functional calculus of self-adjoint operators therefore yields
\begin{equation*}
 e^{-i \eps^{-2}t \H_{\a, \rm diag}^\eps} P^\a_N(\eps) = e^{-i \eps^{-2} t P^\a_N(\eps) \H_{\a}^\eps P^\a_N(\eps)} P^\a_N(\eps),
\end{equation*}
and hence
\begin{equation}\label{estt2}
\Big \| \Big (e^{-i \eps^{-2} t \H_\a^\eps} - e^{-it \eps^{-2} P^\a_N(\eps) \H_\a^\eps P^\a_N(\eps)} \Big )P^\a_N(\eps) \Big \| \lesssim \eps^{N-1} |t|.
\end{equation}
Another application of the triangle inequality allows us to combine this estimate with \eqref{estt1} to infer
\[
\Big \| \Delta(\eps,t) -  \Big (  e^{-it \eps^{-2} P^\a_N(\eps) \H_\a^\eps P^\a_N(\eps)}P^\a_N(\eps) 
- e^{-it \eps^{-2} P_\lambda(\eps) \H^\eps P_\lambda(\eps)} P_\lambda(\eps) \Big ) \Big \| \lesssim \eps + \eps^{N -1}|t|.
\]
In here, we can further rewrite $$P_\lambda(\eps) = P_\lambda(\eps) -P^\a_N(\eps) + P^\a_N(\eps),$$ and note that
\begin{align*}
P^\a_N(\eps) - P_\lambda(\eps) = P^\a_N(\eps) - T^\a_N(\eps) + \sum_{j=1}^N \eps^j (Q_j^\a(\eps) - Q_j(\eps)) - \sum_{j=N+1}^\infty \eps^j Q_j(\eps).
\end{align*}
Recalling \eqref{PTeta} and \eqref{Qest}, we see that
\begin{equation}\label{Pest3}
\| P^\a_N(\eps) - P_\lambda(\eps) \| \lesssim \eps^{N+1} + \eps^{\beta-\alpha},
\end{equation}
which consequently yields
\[
\Big \| \Delta(\eps,t) -  \Big (  e^{-it \eps^{-2} P^\a_N(\eps) \H_\a^\eps P^\a_N(\eps)} - e^{-it \eps^{-2} P_\lambda(\eps) \H^\eps P_\lambda(\eps)} \Big )P^\a_N(\eps) \Big \| \lesssim \eps 
+  \eps^{\beta-\alpha} + \eps^{N -1}|t|.
\]

Finally, we can use formula \eqref{duham} once more to estimate
\begin{align*}
&\, \Big \| \Big (  e^{-it \eps^{-2} P^\a_N(\eps) \H_\a^\eps P^\a_N(\eps)} - e^{-it \eps^{-2} P_\lambda(\eps) \H^\eps P_\lambda(\eps)} \Big )P^\a_N(\eps) \Big \|  \\
& \ \le \int_0^{t \eps^{-2} } \| P^\a_N(\eps) \H_\a^\eps P^\a_N(\eps)  - P_\lambda(\eps) \H^\eps P_\lambda(\eps)\| \, d\tau .
\end{align*}
Here, we can express the difference of operators appearing within the integral via
\begin{align*}
&\, P^\a_N(\eps) \H_\a^\eps P^\a_N(\eps)  - P_\lambda(\eps) \H^\eps P_\lambda(\eps)  \\
& \, = P^\a_N(\eps) \H_\a^\eps P^\a_N(\eps)  - P_\lambda(\eps) \H_\a^\eps P_\lambda(\eps) + \eps^{\beta - \alpha} P_\lambda(\eps) \widehat{W}_\a^{\eps,p} P_\lambda(\eps)\\
&\, = \big(P^\a_N(\eps) - P_\lambda(\eps)\big) \H_\a^\eps P^\a_N(\eps)  + P_\lambda(\eps) \H_\a^\eps \big (P^\a_N(\eps)-P_\lambda(\eps) \big) + \eps^{\beta - \alpha} P_\lambda(\eps) \widehat{W}_\a^{\eps,p} P_\lambda(\eps).
\end{align*}
Estimate \eqref{Pest3}, together with our previous discussion on the boundedness of all the appearing operators then yields
\[
\| P^\a_N(\eps) \H_\a^\eps P^\a_N(\eps)  - P_\lambda(\eps) \H^\eps P_\lambda(\eps)\| \lesssim \eps^{N+1} + \eps^{\beta - \alpha},
\]
and thus
\[
\int_0^{t \eps^{-2} } \| P^\a_N(\eps) \H_\a^\eps P^\a_N(\eps)  - P_\lambda(\eps) \H^\eps P_\lambda(\eps)\| \, d\tau \lesssim \eps^{N-1} |t|+ \eps^{\beta - \alpha -2}|t|.
\]
Recalling that $\beta -\alpha \ge 3$ we infer that
\[
\Big \| \Big (  e^{-it \eps^{-2} P^\a_N(\eps) \H_\a^\eps P^\a_N(\eps)} - e^{-it \eps^{-2} P_\lambda(\eps) \H^\eps P_\lambda(\eps)} \Big )P^\a_N(\eps) \Big \|  \lesssim \eps |t|, \quad \text{for $N\ge 2$.}
\]
Hence, by choosing some fixed $N\ge 2$:
\begin{align*}
\| \Delta (\eps, t)\| \le & \,  \Big \| \Delta(\eps,t) -  \Big (  e^{-it \eps^{-2} P^\a_N(\eps) \H_\a^\eps P^\a_N(\eps)} - e^{-it \eps^{-2} P_\lambda(\eps) \H^\eps P_\lambda(\eps)} \Big )P^\a_N(\eps) \Big \| \\ 
&\, + \Big \| \Big (  e^{-it \eps^{-2} P^\a_N(\eps) \H_\a^\eps P^\a_N(\eps)} - e^{-it \eps^{-2} P_\lambda(\eps) \H^\eps P_\lambda(\eps)} \Big )P^\a_N(\eps) \Big \|\lesssim  \eps (1+|t|),
\end{align*}
which is the desired result. \hfill $\Box$
\medskip

\begin{remark}
It is clear from our proof that the condition that $\a$ vanishes in a neighborhood of the origin can be relaxed to $\a$ being sufficiently small there. More precisely, Theorem \ref{finResing} 
remains valid if
\begin{equation*}
\a(x) \lesssim  
\begin{cases} |x|^{3+\alpha} \quad \text{for $|x| < |x_0|$,} \\
|x|^\beta \quad \quad \text{for $|x|\ge |x_0|$.} \end{cases}
\end{equation*}
The details are left to the interested reader.
\end{remark}

\bigskip

\noindent{\bf Acknowledgements}

The authors want to thank the anonymous referee's for their helpful remarks. C. Sparber gratefully acknowledges financial support by the MPS Simons foundation through award no. 851720. 
The authors have no competing interests to declare that are relevant to the content of this article.

\bigskip


\appendix

\section{Spectral Estimates}\label{sec:specEst}

For the sake of the reader we shall present in here an elementary approach to the proof of exponential decay bounds for the eigenfunctions 
of a wide class of confinement Hamiltonians. Suppose
\[
	H = -\Delta + U(x), \quad \D(H) = C_0^\infty(\R^d),
\]
such that $U: \R^d \to \R$ satisfies
\begin{equation}\label{Vcon}
	U(x) \ge 0, \quad U \in L^\infty_{\text{loc}}(\R^d), \quad \lim_{R \to \infty}\left(\essinf_{x \in \R^d} \big(1-\mathbbm{1}_{B_R(0)}\big)U(x)\right) = \infty,
\end{equation}
where $\mathbbm{1}_{B_R}(x)$ is the indicator function on the ball of radius $R$. As in Lemma \ref{specInfo}, we have that $H$ is essentially self-adjoint with positive, discrete spectrum. Finally, let
\[
	\langle x \rangle := (1+|x|^2)^{1/2}
\]
and note that
\begin{equation}\label{bracket}
	|\nabla\langle x \rangle| \le 1, \quad |\Delta\langle x \rangle| \le d.
\end{equation}
Then, we have the following:
\begin{proposition}\label{ResProp}
Let $K \subset \rho(H)$ be compact, where $\rho(H)$ is the resolvent set of $H$. Then, there exists $\omega_K > 0$ and $M_K < \infty$ such that for all $z\in K$ and $0\le \omega \le \omega_K$, the resolvent satisfies
\begin{equation}\label{resLem1}
	\big\| e^{\omega\langle\, \cdot\, \rangle}(H-z)^{-1}e^{-\omega\langle\, \cdot\, \rangle}\big \| \le M_K.
\end{equation}
\end{proposition}
\begin{proof}
First note that for $f \in C_0^\infty(\R^d)$,
\[
	\Big(e^{\omega\langle\, \cdot\, \rangle}(H-z)^{-1}e^{-\omega\langle\, \cdot\, \rangle}\Big)\Big(e^{\omega\langle\, \cdot\, \rangle}(H-z)e^{-\omega\langle\, \cdot\, \rangle}\Big)f = f,
\]
and so
\begin{equation}\label{InvId}
	e^{\omega \langle\, \cdot \, \rangle}(H-z)^{-1}e^{-\omega\langle\, \cdot\, \rangle} = \Big(e^{\omega\langle\, \cdot\, \rangle}(H-z)e^{-\omega\langle\, \cdot\, \rangle}\Big)^{-1},
\end{equation}
provided $e^{\omega\langle\, \cdot\, \rangle}(H-z)e^{-\omega\langle\, \cdot\, \rangle}$ has a bounded inverse. To show that this is indeed the case, we first compute directly
\[
	e^{\omega\langle\, \cdot\, \rangle}He^{-\omega\langle\, \cdot\, \rangle} = H + \omega D_\omega,
\]
where
\[
	D_\omega := 2\nabla\langle\, \cdot \, \rangle\cdot\nabla + \Delta\langle\, \cdot \, \rangle  - \omega|\nabla \langle\, \cdot \, \rangle|^2.
\]
Then, we have the identity
\[
	e^{\omega\langle\, \cdot\, \rangle}(H-z)e^{-\omega\langle\, \cdot\, \rangle} = (H - z + \omega D_\omega) = \Big(1+\omega D_\omega(H-z)^{-1}\Big)(H-z).
\]
Hence, if we can bound
\begin{equation}\label{aDbound}
	\omega\Vert D_\omega(H-z)^{-1}\Vert \le c < 1
\end{equation}
then $1+\omega D_\omega(H-z)^{-1}$ is invertible and we obtain the desired bounded inverse for $z \in K$:
\begin{equation}\label{InvBound}
	\Big\Vert \Big(e^{\omega\langle\, \cdot\, \rangle}(H-z)e^{-\omega\langle\, \cdot\, \rangle}\Big)^{-1}\Big\Vert \le \frac{1}{1-c}\Vert(H-z)^{-1}\Vert.
\end{equation}

To prove \eqref{aDbound}, we must first bound $D_\omega$: By H\"older's inequality and \eqref{bracket}, we have
\begin{equation*}\begin{split}
	\Vert 2\nabla\langle\, \cdot\, \rangle\cdot\nabla f\Vert_{L^2}^2 & \, = \Big\Vert\sum_{j=1}^d 2\nabla_j\langle\, \cdot\, \rangle\nabla_j f\Big\Vert_{L^2}^2 \\
	& \, \le d\sum_{j=1}^d\Big\Vert 2\nabla_j\langle\, \cdot \, \rangle\nabla_j f\Big\Vert_{L^2}^2\le d\sum_{j=1}^d\Big\Vert 2\nabla_j f\Big\Vert_{L^2}^2.
\end{split}
\end{equation*}
Then, since $U(x) \geq 0$,
\begin{equation*}\begin{split}
	d\sum_{j=1}^d\Big\Vert 2\nabla_j f\Big\Vert_{L^2}^2 = 4d\langle f, -\Delta f \rangle_{L^2} \le 4d\langle f, H f \rangle_{L^2}.
\end{split}\end{equation*}
Finally, let $b \in (0,\infty)$. By Cauchy-Schwarz,
\begin{equation*}\begin{split}
	4d\langle f, H f \rangle_{L^2} \le 4d\Vert bf\Vert_{L^2} \, \Vert \tfrac{1}{b}Hf\Vert_{L^2} \le 2d\Big(b\Vert f\Vert_{L^2}+\tfrac{1}{b}\Vert Hf\Vert_{L^2}\Big)^2,
\end{split}\end{equation*}
and so we have shown
\begin{equation}\label{D1}
	\Vert 2\nabla\langle\, \cdot \, \rangle \cdot \nabla f\Vert_{L^2} \le \sqrt{2d}\Big(b\Vert f\Vert_{L^2}+\tfrac{1}{b}\Vert Hf\Vert_{L^2}\Big).
\end{equation}
The second and third term in $D_\omega$ are easily bounded via \eqref{bracket}, and so with \eqref{D1} we find
\begin{equation}\begin{split}\label{Dfin}
	\Vert D_\omega f\Vert_{L^2} &\le \sqrt{2d}\Big(b\Vert f\Vert_{L^2}+\tfrac{1}{b}\Vert Hf\Vert_{L^2}\Big) + (d+\omega)\Vert f\Vert_{L^2}\\
	&= \frac{1}{b}\sqrt{2d}\, \Vert Hf\Vert_{L^2} + \Big(b\sqrt{2d} + d + \omega\Big)\Vert f\Vert_{L^2}.
\end{split}\end{equation}

This inequality holds in general whenever $f \in \D(H)$, where we now denote $\D(H)$ to be the domain of the self-adjoint extension of $H$. 
Hence, we may take $g \in L^2(\R^d)$ and substitute $f = (H-z)^{-1}g$ into \eqref{Dfin}:
\begin{equation}\label{Dstep}
\begin{split}
	\Vert D_\omega (H-z)^{-1}g\Vert_{L^2} \le & \, \frac{1}{b}\sqrt{2d}\, \Vert H(H-z)^{-1}g\Vert_{L^2} \\ 
	& \, + \Big(b\sqrt{2d} + d + \omega\Big)\Vert (H-z)^{-1}g\Vert_{L^2}.
	\end{split}
\end{equation}
Next, recalling that $K \subset \rho(H)$ is compact, we denote
\[
	\delta_K := \inf_{z\in K} \text{dist}\Big(z,\sigma(H)\Big) > 0, \quad z_K := \sup_ {z\in K}|z| < \infty.
\]
By functional calculus, this gives
\begin{equation}\label{funcBound}
	\sup_{z\in K} \Vert(H-z)^{-1}\Vert \le \frac{1}{\delta_K},\quad \sup_{z\in K} \Vert H(H-z)^{-1}\Vert \le 1+ \frac{z_K}{\delta_K}.
\end{equation}
Returning to \eqref{Dstep}, we choose $b = \sqrt{\delta_k}$ to find
\begin{equation}\label{alphaEst}
	\Vert D_\omega (H-z)^{-1}g\Vert_{L^2} \le \bigg(\sqrt{\frac{2d}{\delta_k}}\bigg(1+ \frac{z_K}{\delta_K}\bigg) + \big(\sqrt{2d\delta_k} + d + \omega\big)\frac{1}{\delta_K}\bigg)\Vert g\Vert_{L^2}.
\end{equation}
Therefore, we can always find an $\omega_K > 0$, such that for all $0 \le \omega \le \omega_K$:
\[
	\omega\Vert D_\omega (H-z)^{-1}g\Vert_{L^2} \le \frac{1}{2}\Vert g\Vert_{L^2},
\]
proving \eqref{aDbound}, and in turn \eqref{InvBound} with $c = \frac{1}{2}$. 

Finally, we apply the estimates \eqref{InvBound} and \eqref{funcBound} within identity \eqref{InvId} to finish the proof:
\begin{equation*}\begin{split}
	\Big \Vert e^{\omega\langle\, \cdot\, \rangle}(H-z)^{-1}e^{-\omega\langle\, \cdot\, \rangle} \Big\Vert = \Big \Vert\Big(e^{\omega\langle\, \cdot\, \rangle}(H-z)e^{-\omega\langle\, \cdot\, \rangle}\Big)^{-1}\Big\Vert 
	\le \frac{2}{\delta_k} := M_K.
\end{split}
\end{equation*}\end{proof}

Next, we show that the eigenfunctions and spectral projections of $H$ decay at worst exponentially:
\begin{proposition}\label{EigenProp}
Let $\lambda \in \R$ be an $m$-degenerate eigenvalue of $H=-\Delta + U$, with associated 
orthonormalized eigenfunction $\chi_k$, $k= 1, \dots, m$. Then, for any $\omega_0 > 0$ there exists $M_{\lambda,\omega_0} < \infty$, such that
\begin{equation}\label{resLem2}
	\big \| e^{\omega_0\langle\, \cdot\, \rangle}\chi_k \big \|_{L^2(\R^d)} \le M_{\lambda,\omega_0}.
\end{equation}
Additionally, for the associated $m$-dimensional spectral projection $P_\lambda$ it holds:
\begin{equation}\label{resLem3}
	\big \| e^{\omega_0\langle\, \cdot\, \rangle}P_\lambda e^{\omega_0\langle\, \cdot\, \rangle}f\big \|_{L^2(\R^d)} \le mM_{\lambda,\omega_0}^2\Vert f\Vert_{L^2(\R^d)}.
\end{equation}
\end{proposition}

\begin{proof} 
Suppose $\lambda \in \sigma(H)$ is an eigenvalue of finite multiplicity $m\in \N$ and $\chi_k$ is an associated (normalized) eigenfunction. For $\nu > 0$, let
\[
\Omega_\nu := \{x \in \R^d \, : \,  U(x) < \lambda + \nu\} ,
\]
and define
\[
	H_\nu := -\Delta_x + (1-\mathbbm{1}_{\Omega_\nu})U + (\lambda+\nu)\mathbbm{1}_{\Omega_\nu},
\]
where $\mathbbm{1}_{\Omega_\nu}$ is the indicator function of $\Omega_\nu$. Then, we can rewrite $H\chi_k = \lambda \chi_k$ via
\[
	(H_\nu - \lambda)\chi_k = (\lambda + \nu - U)\mathbbm{1}_{\Omega_\nu}\chi_k.
\]
By construction, $H_\nu \geq \lambda + \nu > \lambda$, and so $\lambda \in \rho(H_\nu)$, the resolvent set of $H_\nu$. Hence,
\[
	\chi_k = (H_\nu - \lambda)^{-1}(\lambda + \nu - U)\mathbbm{1}_{\Omega_\nu}\chi_k
\]
and we can estimate
\begin{equation}\label{estchi}
\begin{split}
	\big \| e^{\omega_0\langle\, \cdot\, \rangle}\chi_k\big \|_{L^2} &\le \big \| e^{\omega_0\langle\, \cdot\, \rangle}(H_\nu - \lambda)^{-1}e^{-\omega_0\langle\, \cdot\, \rangle}\big \|\\
	&\quad\quad\times \big \| e^{\omega_0\langle\, \cdot\, \rangle}(\lambda + \nu - W)\mathbbm{1}_{\Omega_\nu}\big\|_{L^\infty}\, {\|\chi_k\|}_{L^2}.
\end{split}\end{equation}
For the term involving $(H_\nu - \lambda)^{-1}$, we seek to invoke Proposition \ref{ResProp}.
Note that the latter {\it a-priori} only holds for $\omega_0 \le \omega_K$. 
However from \eqref{alphaEst} it is clear that we may take $\omega_K$ to be arbitrarily large 
by increasing $\delta_K = \inf_{z\in K} \text{dist}\big(z,\sigma(H_\nu)\big)$, where $K \subset \rho(H_\nu)$ is compact. 
This can be done by setting $K = \{\lambda\}$ and noting that $\delta_K \geq \nu$. Hence we may choose $\nu$ large enough such that $\omega_0 \le \omega_K$, and Proposition \ref{ResProp} yields
\[
	\Big\Vert e^{\omega_0\langle\, \cdot\, \rangle}(H_\nu - \lambda)^{-1}e^{-\omega_0\langle\, \cdot\, \rangle} \Big\Vert \le M_K < \infty.
\]
Then, since $\Omega_\nu$ is bounded we have
\[
	{\Vert e^{\omega_0\langle\, \cdot\, \rangle}(\lambda + \nu - W)\mathbbm{1}_{\Omega_\nu}\Vert}_{L^\infty} = 
	\sup_{x \in \Omega_\nu}\big |e^{\omega_0\langle x\rangle}(\lambda + \nu - W(x))\big | \le C_{\lambda,\omega_0} < \infty.
\]
In view of \eqref{estchi}, this yields
\[
	\big \| e^{\omega_0\langle\, \cdot\, \rangle}\chi_k\big \|_{L^2} \le M_KC_{\lambda,\omega_0} \equiv M_{\lambda,\omega_0}.
\]
With \eqref{resLem2} established, we can now obtain an analogous bound for $P_\lambda$:
\begin{equation*}\begin{split}
	\big \|  e^{\omega_0\langle\, \cdot\, \rangle}P_\lambda e^{\omega_0\langle\, \cdot\, \rangle}f\big \|_{L^2} &\le \sum_{k=1}^m\big \| e^{\omega_0\langle\, \cdot\, \rangle}\langle e^{\omega_0\langle\, \cdot\, \rangle}f,\chi_k\rangle_{L^2} \chi_k\big \|_{L^2}\\
	&\le \sum_{k=1}^m M_{\lambda,\omega_0}\big |\langle f,e^{\omega_0\langle\, \cdot\, \rangle}\chi_k\rangle_{L^2}\big |\\
	&\le \sum_{k=1}^m M_{\lambda,\omega_0}\Vert f\Vert_{L^2(\R^d)}\, \big\| e^{\omega_0\langle\, \cdot\, \rangle}\chi_k\big\|_{L^2}\le mM_{\lambda,\omega_0}^2\Vert f\Vert_{L^2}.
\end{split}\end{equation*}

\end{proof}
\begin{remark} 
In general, one expects the eigenfunctions $\chi_k$ to decay even stronger, 
as can be seen from the well-known case of the harmonic oscillator where $\Theta\equiv 1$ and $\alpha =1$. 
However, the proof of such stronger decay properties is usually much more involved, while \eqref{resLem2} is sufficient for our purposes.
\end{remark}

\bibliographystyle{amsplain}

\end{document}